\documentclass[a4paper,preprintnumbers,floatfix,superscriptaddress,pra,twocolumn,showpacs,notitlepage,longbibliography]{revtex4-1}

\usepackage[utf8]{inputenc}
\usepackage[T1]{fontenc}
\usepackage[sc,osf]{mathpazo}
\usepackage{amsmath, amsthm, amssymb,amsfonts,mathbbol,amstext}
\usepackage{graphicx}
\usepackage{dcolumn}
\usepackage{bm}
\usepackage{bbm}
\usepackage{hyperref}
\usepackage{mathtools}
\usepackage{comment}
\usepackage{color}
\usepackage{multirow}
\usepackage{pgf,tikz}
\usepackage{mathrsfs}
\usepackage{physics}
\usetikzlibrary{arrows}
\usepackage{soul,xcolor}
\usepackage{tikz}
\usepackage{adjustbox}
\DeclareMathOperator{\spn}{Span}

\newtheorem{theorem}{Theorem}
\newtheorem{definition}{Definition}

\newcommand\myeq{\stackrel{\mathclap{\normalfont\mbox{$\epsilon$}}}{=}}

\begin{document}

\title{Preprocessing operations and the reverse compression}
\date{\today}

\author{Matheus Capela}
\email{matheus@qpequi.com}
\affiliation{Institute of Physics, Federal University of Goi\'{a}s, POBOX 131, 74001-970, Goi\^{a}nia, Brazil}

\author{Fabio Costa}
\email{f.costa@uq.edu.au}
\affiliation{Centre for Engineered Quantum Systems, School of Mathematics and Physics, The University of Queensland, St Lucia, QLD 4072, Australia}

\begin{abstract}
The task of compression of data -- as stated by the source coding theorem -- is one of the cornerstones of information theory. Data compression usually exploits statistical redundancies in the data according to its prior distribution. Motivated by situations where one does not have access to the statistics of data, but has some information about a transformation that is going to be applied to it, we propose a novel method for data compression called reverse compression. It is defined in such a way that works for both classical and quantum information processes, and furthermore relies exclusively on the channel to be used: all input data leading to indistinguishable outputs is compressed to the same state, regardless of their prior distribution. Moreover, this process can be characterized as a higher order operation within the type of preprocessing. We also consider as an example the application of the method to the classical and quantum erasure channel. The examples suggest that noiseless reverse compression can take place only in trivial cases, although meaningful instances of noisy reverse compression can exist.
\end{abstract}

\maketitle


\section{Introduction}
 
Information theory is ubiquitous feature of modern science. Its cornerstone is Shannon's mathematical theory of communication \cite{shannon48}, dealing with two main problems: channel coding and source coding. Channel capacity, defined in the former as the amount of information that can be reliably sent through a channel, is shown to be equal to the mutual information between the input and output random variables maximized over all possible input probability distributions. Thus, it relies only on the channel features. In the latter,  the concept of Shannon entropy characterizes the limit of how much a random source can be reliably compressed. Since then, many links with other branches of science have been achieved. For instance, it allowed a deep connection with statistical mechanics through the Jaynes' maximum entropy principle \cite{jaynes1957information}, setting the equilibrium probability distribution associated with a thermodynamical (statistical) system as the one maximizing Shannon entropy respecting the relevant physical constraints. 

We address here the problem of compression of a source without relying on its statistical features, but only on the knowledge of properties of the channel in which the messages are to be send through. The practical situation of interest can be thought, for instance, as the task of transmitting a picture through some known processing algorithm or device, without any prior knowledge of the statistical distribution of the pixels composing the picture. For example, the device could be some analysis tool, returning a set of properties of the picture. By knowing in advance that the device has little sensitivity to some aspects of the picture (e.g., the output might not depend much on certain differences in color), we can save communication bandwidth by compressing the image before sending it to the device. As we want the protocol to be independent of the input distribution, Shannon's compression method does not help towards this task.

We call such a procedure reverse compression of a source with respect to a channel. We formalize this preprocessing operation for both classical and quantum information processes. We also show examples where compression is possible given some finite noise. However, it appears that no noiseless compression can be achieved, except for trivial examples.

The reverse compression protocol introduced here can be understood as a particular case of superchannel operation, which is a problem dealt with by Shannon in the early days of Information Theory. His interest was in the pre-ordering of classical channels. He was able to show that if a channel can simulate another one by preprocessing and postprocessing with shared randomness, and there is a code for the latter, then there is a code at least as efficient for the former \cite{shannon1958note}. Much development has been achieved in this problem relating the characterization of specific classes of channels \cite{zhang2013analytical,nasser2018characterizations,haddadpour2016simulation,d2019distance}. Nevertheless, it remains as an open area of research. An area of statistical analysis which is a particular case of the one considered by Shannon is the Theory of Comparison \cite{blackwell1953equivalent}. It consists of the simulation of channels with the same input space by postprocessing only. Recent developments and extension to the Quantum Information Theory have been achieved in \cite{jenvcova2016comparison,buscemi2017comparison}. On the other hand, it is also possible to define relevant information processing tasks involving exclusively preprocessing, for instance, the reverse compression presented here. The authors in \cite{buscemi2005clean} have considered the preprocessing of Positive Operator Value Measures (POVM).

Recent interest has been shown in the quantum theory of superchannels. For instance, Ref.~\cite{gour2019comparison} has considered the problem of quantum channel simulation by superchannels with preprocessing and postprocessing with quantum side information, and the extension of the notion of entropy for channels. With a more physical motivation, the authors in \cite{faist2019thermodynamic} have considered the problem of quantum channel simulation by so called thermal operations, and derived that in the asymptotic limit the thermodynamical cost of the simulation is given by a unique measure: the difference in thermodynamical capacity of the correspondent channels. Furthermore, superchannels extend natururally to more general higher order transformations, which are central in the study of quantum causal structures \cite{chiribella2008transforming, chiribella09b, oreshkov12, araujo15, oreshkov15, costa2016, Perinotti2017, Giarmatzi2018, Bisio2019, Kissinger2019}.

The paper is organized as follows. In section \ref{notation} we define the notation to be used along the text. In section \ref{shannon} we discuss the Shannon compression method in order to make clear the distinction with the preprocessing task considered here. In section \ref{classical} we define the reverse compression of data for memoryless classical channels used without feedback. The exposition starts from the case of single-shot uses of the channel, leading to the asymptotic setting, with arbitrarily many uses. In section \ref{quantum} we consider the definition of reverse compression of data for quantum channels. As a general rule we consider as an application the case of the erasure channel. Finally, in section \ref{conclusion} we make our final remarks and discussions.

\section{Notation and Preliminaries} \label{notation}

We denote random variables by English upper case letters, their outcomes by lower case letters, and alphabets -- the collection of all possible outcomes, also called the values of a random variable -- by calligraphic letters. For instance, $X$ represents a random variable with outcomes denoted by $x$, while $\mathcal{X}$ represents its alphabet. In what follows, we only consider alphabets with a finite number of outcomes denoted as $|\mathcal{X}|$. More generally, the symbol $|\cdot|$ represents the cardinality of a set.

A probability distribution of the random variable $X$, denoted as $P(X)$, is an assignment of real numbers to each outcome in its alphabet, written as $P(x)$ and called probability masses. The probability masses satisfy $0 \leq P(x) \leq 1$ for any outcome, and $\sum_{x\in \mathcal{X}} P(x)=1$. The expected value of a real valued random variable $R$ is defined as $E(R) \coloneqq \sum_{r\in \mathcal{R}}P(r)r$ and its variance as $\mathrm{var}(R) \coloneqq E([R-E(R)]^2)$. A sequence of $k$ random variables $X^k \coloneqq (X_1,\cdots,X_k)$ taking values in the same alphabet $\mathcal{X}$ is said to be independent and identically distributed (i.i.d.) if it is drawn according to a probability of the form $P(X^k) = \prod_{i=1}^k P(X_i)$, where $P(X_i)=P(X_j)$ for any $1 \leq i,j \leq k$. We define the sample average of a sequence of real-valued random variables as $\overline{R^k} \coloneqq \frac{1}{k} \sum_{i=1}^k R_i$. Note the sample average is defined here as a random variable assigning the arithmetic average to each outcome $r^k$. Additionally, we say the sequence of real-valued joint random variables $( R^k )_{k \in \mathbb{N}}$, each drawn according to $P(R^k)$, converges in probability to a real number $\zeta$ if $\Pr \left[ |\overline{R^k}-\zeta| \leq \delta \right] \rightarrow 1$ as $k \rightarrow \infty$ for arbitrary positive real number $\delta$, and write $R^k \xrightarrow{P} \zeta$ in such case.

The fidelity of two probability distributions $P$ and $Q$ of the same random variable $X$ is defined as 

\begin{equation} \label{classfid}
F\left[P(X),Q(X)\right]=\left[\sum_{x}\sqrt{P(x)Q(x)}\right]^2.
\end{equation}
It is a quantity regarded as a measure of distinguishability of probability distributions such that $0\leq F(P,Q)\leq 1$, and $F(P,Q)=1$ if and only if $P=Q$.

A communication system consists of a sender encoding messages emitted by a source to be transmitted through a channel -- the physical medium allowing communication between distant parties -- and decoded by a receiver. The channel is described by a probability distribution of the output variable $Y$ conditional to the input variable $X$, and denoted by $P(Y|X)$. Moreover, the channel can be understood as a process that transforms input probability distributions $P(X)$ into output ones $P(Y)$ by the relation $P(Y) = \sum_{x \in \mathcal{X}} P(Y|x)P(x)$.

The quantum analogue of the alphabet of a random variable is the space of pure quantum states. Therefore, to each random variable $X$ with alphabet $\mathcal{X}$ we assign a collection of unit orthogonal vectors $\{\ket{x}\}_{x \in \mathcal{X}}$, and define the associated Hilbert space by any linear composition of them denoted by $\mathcal{H}_{X}\coloneqq \spn\ \{\ket{x}\}_{x \in \mathcal{X}}\cong \mathbb{C}^{|\mathcal{X}|}$. We denote the space of linear operators acting on a vector space $\mathcal{H}$ by $\mathcal{L}(\mathcal{H})$, and the set of states of the corresponding quantum system -- positive linear operators with unit trace -- as $\mathcal{D}(\mathcal{H})$.

A linear map $\mathcal{M}^{X\rightarrow Y}:\mathcal{L}(\mathcal{H}_{X})\rightarrow \mathcal{L}(\mathcal{H}_{Y})$ is called completely positive (CP) if its extension $\mathcal{I}^{R} \otimes \mathcal{M}^{X\rightarrow Y}$ to any finite dimensional system $R$ is positive, where $\mathcal{I}^{R}:\mathcal{L}(\mathcal{H}_{R})\rightarrow \mathcal{L}(\mathcal{H}_{R})$ is the identity map. We say $\mathcal{M}^{X\rightarrow Y}$ is trace preserving (TP) if $\Tr_{Y}\circ \mathcal{M}^{X\rightarrow Y} = \Tr_{X}$. A quantum channel is a linear completely positive and trace preserving (CPTP) map. We refer the reader to \cite{wilde2011classical} for a more complete discussion on quantum channels.

The fidelity of the quantum states $\rho$ and $\sigma$ is defined here as $F(\rho,\sigma)\coloneqq \left[\Tr(\sqrt{\sqrt{\rho}\sigma\sqrt{\rho}})\right]^{2}$, following the original definition in \cite{jozsa}. If the states considered are diagonal in the same basis, assuming the form $\rho=\sum_x P(x)\ket{x}\bra{x}$ and $\sigma=\sum_x Q(x)\ket{x}\bra{x}$, then Equation (\ref{classfid}) is recovered. 

\section{Overview on Shannon Compression} \label{shannon}
 
The compression of a source is defined by the task of minimizing redundancy of the messages to be sent through a channel relying on its probability distribution, and thus maximizing the efficiency of the communication protocol. This section is based on \cite{cover,csiszarkorner}.

Shannon`s compression theorem states that there is an encoding procedure that is asymptotically reliable, and the goal of this section is to make mathematically precise such terminology. Shannon's compression method is based on the law of large numbers stating that a sequence of i.i.d real valued random variables $R^k$ with finite variance converges in probability to its expected value, $\overline{R^k} \xrightarrow{P} E(R)$. Taking $R_k= \log P(X_i)$, the law of large numbers implies the convergence $\overline{\log P(X^k)} \xrightarrow{P} H(X)$, which means that for any positive real number $\delta$ we have

\begin{equation} \label{largenumbers}
\Pr \left[ \left| \frac{1}{k} \log P(X_1,\cdots,X_k) - H(X) \right|\leq \delta  \right] \rightarrow 1 
\end{equation}
as $k \rightarrow \infty$. 

It suggests we can take the set of likely sequences as consisting of the so called $\delta$-typical sequences defined by the property $|\overline{\log P(X^k)} - H(X)|\leq \delta$, where $\delta$ can be made as any arbitrarily small positive real number. The set of all $\delta$-typical sequences of length $k$, for which the probability can be made arbitrarily close to unit for sufficiently large $k$, is denoted here by $T_{\delta}^{(k)}$. Moreover, we can encode only the typical sequences and neglect the atypical ones with probability of decoding a sequence wrong arbitrarily close to zero, as it follows from Equation (\ref{largenumbers}) that $\Pr\left[ \mathcal{X}^{k}\setminus T_{\delta}^{(k)}\right] \rightarrow 0$ as $k \to \infty$. 

It can be shown that the number of typical sequences one needs to encode ranges as $(1-\delta)2^{\left[k(H(X)-\delta) \right]} \leq | T_{\delta}^{(k)} | \leq 2^{\left[k(H(X)+\delta) \right]}$, and for $\delta \rightarrow 0$ one needs to encode $| T_{0}^{(k)} | = 2^{kH(X)}$ sequences with probability of error arbitrarily close to zero if $k$ is large enough ($k\to \infty$). Note that $H(X)$ may not be an integer and therefore $| T_{0}^{(k)} |$ should be replaced in the equation above by the smallest integer greater than or equal to it, but we keep this notation for simplicity. It means that we can reliably compress a variable by encoding only $2^{kH(X)}$ sequences, rather than encoding all the $\vert \mathcal{X} \vert^{k}$. This is often called as the direct part of the Shannon compression theorem: the proof of existence of a compression-decompression scheme 

\begin{equation}
\left(f:\mathcal{X}^{k}\rightarrow \{0,1\}^{kH(X)},g: \{0,1\}^{kH(X)} \rightarrow \mathcal{X}^{k} \right) \nonumber
\end{equation} 
with arbitrarily small probability of error defined by $\Pr\left[ g\circ f (X^k)\neq X^k \right]$ for sufficiently large $k$. The only case where there is no advantage in this encoding method is when all outcomes are uniformly distributed, and therefore all the sequences are typical. The converse part of the Shannon compression theorem states that for any compression-decompression scheme $\left(f:\mathcal{X}^{k}\rightarrow \{0,1\}^{n},g: \{0,1\}^{n} \rightarrow \mathcal{X}^{k} \right)$ with a compression rate $\frac{n}{k}$ less than $H(X)$ has an arbitrarily large probability of error.

\section{Classical Reverse Compression} \label{classical}

\subsection{Single-Shot Settings}

We start with a definition of preprocessing for single-shot scenarios, i.e., information processing protocols with a single use of the channel. For single-shot settings the definition of reverse compression takes a simpler and clear form, and furthermore its extension to asymptotic scenarios is an immediate application considered subsequently.

The question addressed here is the possibility of replacement of the input variable of a classical channel by a less redundant one, as it is the case for the Shannon compression method, but without relying on the knowledge of the input variable distribution $P(X)$ nevertheless. Instead, the idea is to look for redundancies in how the input is mapped to the output: if two input variables produce very similar output distributions, then they can be replaced by a single one. Importantly, here one compresses the input based on information of a transformation that will take place in the future, hence the terminology \textit{reverse compression}. Moreover, the main difference between the Shannon compression method and ours lies in the task of communication: the goal of Shannon compression is to communicate a message through a noiseless channel with arbitrarily small probability of error, while reverse compression is defined in the following by channel indistinguishability after preprocessing of data.

Two probability distributions $P$ and $Q$ are equal if and only if their fidelity is equal to unity, $F(P,Q)=1$. Allowing for slightly distinguishable events, we say that probability distributions $P$ and $Q$ are $\epsilon$-indistinguishable if $F[P,Q] \geq 1-\epsilon$, and denote it as $P \myeq Q$.
The reverse compression of a source with respect to a channel consists of replacing the input random variable $X$ by a random variable $Z$ with smaller alphabet such that the output variable $Y$ conditional distribution remains approximately the same. This stage is called here compression. The appropriate mapping from $\mathcal{Z}$ into $\mathcal{X}$ (necessary to feed the variable into the desired channel) is called decompression. Figure \ref{diagram} shows a diagram representing the process.  

\begin{figure}   
\begin{adjustbox}{width=0.4\textwidth}
\begin{tikzpicture}
\draw (5,0) -- (5,1.5); \node[scale=3] at (6,0.75) {$X$};
\node[scale=2.5] at (8.6,2.4) {$\mathcal{X}$};

\draw (5,1.5) -- (2,2); \draw (5,1.5) -- (3,2);
\draw (5,1.5) -- (8,2); \draw (5,1.5) -- (7,2);
\draw (2,2) -- (2,3); \draw (3,2) -- (3,3); 
\node[scale=2] at (2.5,2.5) {$x_1$}; \node[scale=2] at (3.5,2.5) {$\widehat{x}_1$};
\node[scale=2] at (5,2.5) {$\cdots$};
\draw (8,2) -- (8,3); \draw (7,2) -- (7,3); \node[scale=2] at (2.5,2.5) {$x_1$};
\node[scale=2] at (6.5,2.5) {$\widehat{x}_n$}; \node[scale=2] at (7.5,2.5) {$x_n$};

\draw (1,3) rectangle (9,4);
\node[scale=2] at (5,3.5) {$\lambda$};

\draw (2,3) -- (2,4); \draw (3,3) -- (2,4);
\draw (8,3) -- (8,4); \draw (7,3) -- (8,4);

\draw (2,4) -- (2,5); \draw (8,4) -- (8,5);
\node[scale=2] at (3,4.5) {$z=1$}; \node[scale=2] at (7,4.5) {$z=n$};
\node[scale=2] at (5,4.5) {$\cdots$};
\node[scale=2.5] at (8.6,4.45) {$\mathcal{Z}$};

\draw (1,5) rectangle (9,6);
\node[scale=2] at (5,5.5) {$\varphi$};

\draw (2,5) -- (2,7); \draw (8,5) -- (8,7);
\draw (2,7) -- (5,7.5); \draw (8,7) -- (5,7.5);
\node[scale=2] at (2.5,6.5) {$x_1$}; \node[scale=2] at (7.5,6.5) {$x_n$};
\node[scale=2] at (5,6.5) {$\cdots$};

\draw[red,thick,dashed] (0,1.5) rectangle (10,7.5);

\node[scale=3] at (6,8.25) {$X'$};
\node[scale=2.5] at (8.6,6.45) {$\mathcal{X}'$};

\draw (5,7.5) -- (5,9);
\draw (3,9) rectangle (7,12);
\draw (5,12) -- (5,13.5);

\node[scale=3] at (5,10.5) {$P(Y|X')$};
\node[scale=3] at (5.9,13) {$Y$};
\node[scale=3, red] at (0.4,8) {$\Lambda$};

\end{tikzpicture}
\end{adjustbox}
    \caption{Diagram representing the reverse compression of a channel. The overall process inside the dashed box represents the reverse compression $\Lambda$. It consists of a compression operation $\lambda:\mathcal{X} \rightarrow \mathcal{Z}$ and a decompression operation $\varphi: \mathcal{Z} \rightarrow \mathcal{X}'$. In the figure we have $\mathcal{X}'$ defined as a copy of $\mathcal{X}$.}
      \label{diagram}

\end{figure}
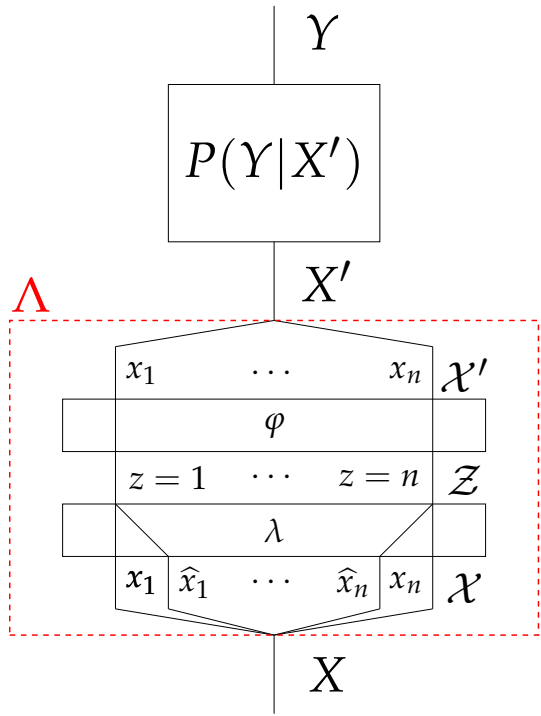

Consider a partition of the input variable's alphabet according to the following criteria. We take sets $\mathcal{A}_z$ of outcomes of $X$ such that, for some arbitrarily fixed $\epsilon > 0$, the conditional probabilities $P\left(Y|X=x\right)$ and $P\left(Y|X=x'\right)$ are $\epsilon$-indistinguishable, i.e.,

\begin{equation}
F \left[ P\left(Y|X=x\right) , P\left(Y|X=x'\right) \right]\geq 1 - \epsilon,
\end{equation}
for any $x$ and $x'$ in $\mathcal{A}_z$.

A collection of non-empty sets $\mathcal{A}_z$ is a partition of the input alphabet $\mathcal{X}$ if they are pair-wise disjoint and they cover $\mathcal{X}$, namely $\cup_z \mathcal{A}_z = \mathcal{X}$. The mapping $\mathcal{A}_z \ni x \mapsto z $ defines a collective random variable $Z$, interpreted as the compressed input. The channel from $Z$ to $Y$ is given by

\begin{equation} \label{decoding}
P(Y|Z=z) = \sum_{x\in \mathcal{A}_z} P(Y|X=x)P(X=x|Z=z),
\end{equation}
where $P(X|Z=z)$ is any distribution with support in $\mathcal{A}_z$, that is, $P(X=x|Z=z)=0$ if $x\not \in \mathcal{A}_z$. 

Given the definition above, we have, for each $x \in \mathcal{A}_z$
\begin{equation}
F\left[ P(Y|X=x) , P(Y|Z=z) \right] \geq 1-\epsilon.
\end{equation}
Note that as the collection $\mathcal{P}_{\epsilon} = \{ \mathcal{A}_z \}_{z \in \mathcal{Z}}$ is a partition of the input alphabet, the random variable $Z$ has a smaller (or equal) cardinality than $X$. It is worth to note that a partition $\mathcal{P}=\{\mathcal{A}_z\}_{z \in \mathcal{Z}}$ of $\mathcal{X}$ is equivalent to a surjective mapping $\lambda:\mathcal{X} \rightarrow \mathcal{Z}$. Each map $\lambda$ of that type corresponds to a partition with elements defined as $\mathcal{A}_z \coloneqq \lambda^{-1}(z)$, where $\lambda^{-1}(z)\coloneqq \{x \in \mathcal{X} | \lambda(x)=z \}$ is the inverse image of the set $\{z\}$ under the mapping $\lambda$. Moreover, a partition is equivalent to an appropriate mapping on the input alphabet. The partition $\mathcal{P}_{\epsilon}$ is called a reverse compression. The channel $P(X|Z)$ in Eq.~\eqref{decoding} is a decompression. For definiteness, we can choose an arbitrary but fixed $x_z\in \mathcal{A}_z$ for each subset and define the decompression as the mapping from $z$ to this $x_z$: $P(X=x|Z=z)=\delta_{x x_z}$. See Figure \ref{venndiagram}. A decompression with such a definition is clearly not unique. It follows the formal definition. Let $0 \leq \epsilon \leq 1$.

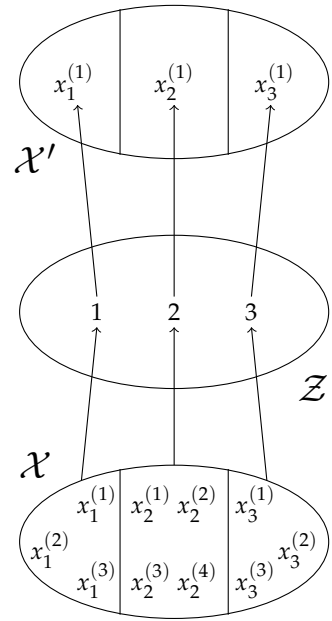
\begin{figure}   
\begin{adjustbox}{width=0.25\textwidth}
\begin{tikzpicture}
  \draw (0,0) ellipse (2cm and 1cm);
  	\draw (-0.7,0.95) -- (-0.7,-0.95);
  	\draw (0.7,0.95) -- (0.7,-0.95);
	\node[scale=1] at (-1,0.5) {$x_1^{(1)}$};
	\node[scale=1] at (-1.6,-0.1) {$x_1^{(2)}$};
	\node[scale=1] at (-1,-0.5) {$x_1^{(3)}$};	
	\node[scale=1] at (-0.3,0.5) {$x_2^{(1)}$};	
	\node[scale=1] at (0.3,0.5) {$x_2^{(2)}$};	
	\node[scale=1] at (-0.3,-0.5) {$x_2^{(3)}$};	
	\node[scale=1] at (0.3,-0.5) {$x_2^{(4)}$};	
	\node[scale=1] at (1.05,0.5) {$x_3^{(1)}$};
	\node[scale=1] at (1.6,-0.1) {$x_3^{(2)}$};
	\node[scale=1] at (1.05,-0.5) {$x_3^{(3)}$};	
	\draw [->] (-1.2,0.80) -- (-1,2.8);	\draw [->] (0,1) -- (0,2.8); \draw [->] (1.2,0.8) -- (1,2.8);
  \draw (0,3) ellipse (2cm and 1cm);
  	\node[scale=1] at (-1.3,6) {$x_1^{(1)}$};	
    \node[scale=1] at (0,6) {$x_2^{(1)}$};
    \node[scale=1] at (1.3,6) {$x_3^{(1)}$};
  \draw (0,6) ellipse (2cm and 1cm);
  	\node[scale=1] at (-1,3) {$1$};	
    \node[scale=1] at (0,3) {$2$};
    \node[scale=1] at (1,3) {$3$};
	\draw [->] (-1,3.2) -- (-1.25,5.7);	\draw [->] (0,3.2) -- (0,5.7); \draw [->] (1,3.2) -- (1.25,5.7);
	\node[scale=1.5] at (-1.8,5) {$\mathcal{X}'$};
	\node[scale=1.5] at (-1.8,1) {$\mathcal{X}$};
	\node[scale=1.5] at (1.8,2) {$\mathcal{Z}$};
  	\draw (-0.7,6.95) -- (-0.7,5.05);
  	\draw (0.7,6.95) -- (0.7,5.05);
\end{tikzpicture}
\end{adjustbox}
    \caption{Venn diagram representing the reverse compression operation.}
      \label{venndiagram}

\end{figure}

\begin{definition}[Single-Shot Reverse Compression] \label{RC} An $\epsilon$-reverse compression of a channel $P(Y|X)$ is the smallest partition $\mathcal{P}_{\epsilon}$ such that each pair $x$ and $x'$ belonging to the same element of $\mathcal{P}_{\epsilon}$ satisfies $F[P(Y|x),P(Y|x')]\geq 1-\epsilon$. 
\end{definition}

The quality of the reverse compression is clearly quantified by the cardinality of the variable $Z$, and for this matter we define the \emph{compressibility} of a channel as

\begin{equation}
\Gamma_{\epsilon}[P(Y|X)] \coloneqq \frac{|\mathcal{X}|-|\mathcal{P}_{\epsilon}|}{|\mathcal{X}|-1}.
\end{equation}
The compressibility defines the amount of redundancy in the input, relative to the channel $P(Y|X)$, and it is clearly a monotonically increasing function of $\epsilon$. It also satisfies $0 \leq \Gamma_{\epsilon}[P(Y|X)] \leq 1$, as the worst possible reverse compression consists in assigning a set to each possible input outcome and the best reverse compression is given by collecting all input outcomes into a single set. 

Two limit cases can be readily identified. For a channel that outputs a variable independent of the input, $P(Y|X=x)=P(Y|X=x')$ for all $x$, $x'$, maximal compression is possible: we can replace all input variables with a single, arbitrary one without any effect on the output statistics, resulting in $\Gamma_{\epsilon}=1$ for all $\epsilon$. On the other end of the spectrum, an identity channel $P(Y=y|X=x)=\delta_{y\,x}$ does not tolerate any compression: for any two distinct inputs $x\neq x'$, the output distributions are perfectly distinguishable, $F[P(Y|x),P(Y|x')]=0$, which means that, for any $\epsilon > 0$, the only $\epsilon$-reverse compression is the partition that assigns each $x\in\mathcal{X}$ to a distinct set, resulting in $\Gamma_{\epsilon}=0$.

\subsection{Asymptotic Settings}
 
One can always take the asymptotic regime as a particular case of the single-shot setting mentioned above by considering the channel acting on infinitely many input variables, thus producing as many output ones. The channel $P(Y|X)$ in the asymptotic regime is a conditional distribution of the type $P(Y_1,Y_2,\cdots|X_1,X_2,\cdots)$. The reverse compression is identified in this case as a global operation on the joint input random variables, as illustrated in Figure \ref{asymptotic}.

\begin{figure}   
\begin{adjustbox}{width=0.4\textwidth}
\begin{tikzpicture}
\draw (1.5,0)--(1.5,-1); \draw (7.5,0)--(7.5,-1); 
\node[scale=2] at (4.5,-0.5) {$\cdots$};
\node[scale=2] at (2,-0.5) {$X_1$};\node[scale=2] at (8,-0.5) {$X_k$}; 
\draw[red,thick] (0,0) rectangle (9,2); 
\node[scale=2] at (4.5,1) {$\Lambda_k$};
\draw (1.5,2)--(1.5,3); \draw (7.5,2)--(7.5,3); 
\node[scale=2] at (2,2.5) {$X_1'$};\node[scale=2] at (8,2.5) {$X_k'$}; 
\node[scale=2] at (4.5,4) {$\cdots$}; 
\draw[thick] (0,3) rectangle (3,5); \draw[thick] (6,3) rectangle (9,5);
\node[scale=2] at (1.5,4) {$P(Y_1|X_1')$};\node[scale=2] at (7.5,4) {$P(Y_k|X_k')$}; 
\draw (1.5,5)--(1.5,6); \draw (7.5,5)--(7.5,6); 
\node[scale=2] at (2,5.4) {$Y_1$};\node[scale=2] at (8,5.4) {$Y_k$}; 

\end{tikzpicture}
\end{adjustbox}
    \caption{The reverse compression of $k$ independent uses of a channel acts as a joint operation on the input of the $k$ input variables. It is a mapping of the type $\Lambda_k:\mathcal{X}^{k}\rightarrow\mathcal{X}'^{k}$.}
      \label{asymptotic}

\end{figure}
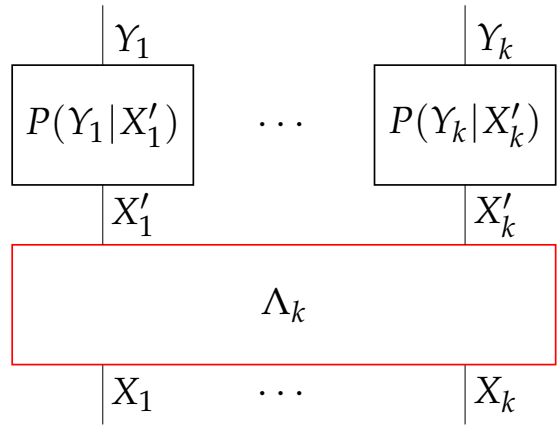

In order to use the previous definition of reverse compression for an asymptotic scenario we need to define a partition with the same properties as above but defined on the space of input sequences. For doing so, consider $k$ copies of the memoryless channel used without feedback, described by 

\begin{equation}
    P(Y^k|X^k)=\prod_{i=1}^{k} P(Y_i|X_i).
\end{equation}

For a channel with such properties, the fidelity of the output distribution conditional on two different inputs takes the simple product form

\begin{equation} \label{eq:fidelityprod}
F[P(Y^k|x^k),P(Y^k|\widehat{x}^k)] = \prod_{i=1}^{k} F[P(Y_i|x_i), P(Y_i|\widehat{x}_i)]
\end{equation}
where $x^k$ and $\widehat{x}^k$ are two arbitrary input sequences. It is clear that $F[P(Y^k|x^k),P(Y^k|\widehat{x}^k)]=1$ if and only if each term $F[P(Y_i|x_i), P(Y_i|\widehat{x}_i)]=1$. A proof of Equation (\ref{eq:fidelityprod}) can be found in the Appendix \ref{fidelity}. It yields the definition of reverse compression for asymptotic settings.

\begin{definition}[Asymptotic Reverse Compression]
An asymptotic $\epsilon$-reverse compression of a memoryless channel used without feedback is a sequence of partitions $\left( \mathcal{P}_{\epsilon}^{(k)} \right)_{k \in \mathbb{N}}$ such that each $\mathcal{P}_{\epsilon}^{(k)}$ is an $\epsilon$-reverse compression of the channel $P(Y^k|X^k)=\prod_{i=1}^{k}P(Y_i|X_i)$.
\end{definition}

The above definition clearly does not define a unique sequence of partitions, but any $\epsilon$-reverse compression is  equivalent for the purpose of the compression task. In the asymptotic regime, we quantify the quality of the method by the asymptotic compressibility defined as

\begin{equation}
\Delta_{\epsilon}[P(Y|X)] \coloneqq \lim_{k \rightarrow \infty} \Gamma_{\epsilon}^{(k)}[P(Y|X)],
\end{equation}
where $\Gamma_{\epsilon}^{(k)}[P(Y|X)]$ is the compressibility of $k$ independent uses of the channel $P(Y|X)$. Similarly to the single-shot case, the asymptotic compressibility is a monotonically increasing function of $\epsilon$. Therefore, the compressibility could have been used in the definition of reverse compression. Alternatively we could have defined it requesting the partition $\mathcal{P}_{\epsilon}$ to have maximum compressibility.

\subsection{Examples} \label{examples}

\subsubsection{The erasure channel}
 
The study case considered here is the erasure channel, outputting the erasure symbol denoted as $\alpha$ with probability $\eta$, and transmitting correctly the input symbol emitted by the source with probability $1-\eta$. If the input space is $\mathcal{X}=\{1,\cdots,r\}$, then the output space is $\mathcal{Y}=\{1,\cdots,r,\alpha\}$. The channel is described by the conditional probability distribution

\begin{equation}
    P_{\eta}(y|x)=(1-\eta) \delta_{x,y} + \eta \delta_{\alpha,y},
\end{equation}
with $0 \leq \eta \leq 1$. 

For the erasure, channel we have (see Appendix \ref{reversefidEC})

\begin{equation}
 F[P_{\eta}(Y^k|x^k),P_{\eta}(Y^k|\widehat{x}^k)]=\eta^{2 S(x^k,\widehat{x}^k)},
\end{equation}
where $S(x^k,\widehat{x}^k)$ is the number of different components of the sequences $x^k$ and $\widehat{x}^k$. For the problem considered here, we want to define the optimal partition -- or appropriate bounds on it -- such that, for each pair of input sequences $x_{z}^k,\widehat{x}_{z}^k$ belonging to the same element of the partition, we have $\eta^{2 S(x^k,\widehat{x}^k)} \geq 1-\epsilon$. Then, for fixed $0<\epsilon,\eta<1$ we must have

\begin{equation}
S(x^k,\widehat{x}^k) \leq \frac{1}{2}\frac{\log(1-\epsilon)}{\log \eta}.
\end{equation}

Consider the case where $\eta=1/2$, i.e., in each use of the channel the input variable $x$ is perfectly mapped into the variable $y=x$ with probability $1/2$ or it is mapped to an error variable $y=\alpha$ with probability $1/2$. In this case, the fidelity threshold reduces to $(0.25)^{ |S(x^k,\widehat{x}^k)|} \geq 1-\epsilon$. Just to be clear, the best partition one could do is $\{\mathcal{X}^k\}$ and the worst possible partition would be $\{\{x^k\}\}_{x^k\in \mathcal{X}^k}$. Let us consider the partition $\{ \{x^k,\widehat{x}^k\}, \{w^k\} \}_{w^k \in \mathcal{X}^k \setminus \{x^k,\widehat{x}^k\}}$, which has $|\mathcal{X}^k| - 1$ elements. The best we could do is to take $x^k$ and $\widehat{x}^k$ with only one different component such that $|S(x^k,\widehat{x}^k)|=1$, for instance $x^k=(0,0,\cdots,0)$ and $\widehat{x}^k=(1,0,\cdots,0)$. In this case, we should allow for an error $\epsilon \geq 0.85$ in fidelity. Thus we cannot compress input variables according to the reverse compression method with arbitrarily high fidelity for the erasure channel. The following example considers an example of asymptotic process with positive compressibility.

A possible way to make sure the reverse fidelity of two $k$-sequences with respect to the channel $P(Y^k|X^k)$,
\begin{equation}
\widetilde{F}(x^k,\widehat{x}^k)\coloneqq F[P(Y^k|x^k),P(Y^k|\widehat{x}^k)],
\end{equation}
is arbitrarily large is exploiting the factorization property stated in the Equation (\ref{eq:fidelityprod}), i.e., if $\widetilde{F}(x_i,\widehat{x}_i)\geq (1-\epsilon)^{\frac{1}{k}}$, then $\widetilde{F}(x^k,\widehat{x}^k) \geq 1-\epsilon$.

Inspired by the erasure channel, consider the case where one can take partitions of the input space -- with respect to the reverse compression -- such that each sequence belonging to the same element of the partition has no more than $s$ different components. If $\mathcal{P}_{\epsilon}^{(k)}$ is a reverse compression of the channel $P(Y^k|X^k)$, then $S(x^k,\widehat{x}^k) \leq s$ for each pair $ x^k,\widehat{x}^k $ belonging to the same element of the partition. We must have then 

\begin{equation}
|\mathcal{P}_{\epsilon}^{(k)} | \leq |\mathcal{X}|^{k-s}. 
\end{equation}
We have furthermore the positive asymptotic compressibility $\Delta_{\epsilon} \geq 1-|\mathcal{X}|^{-s}$. In Appendix \ref{Conjecture}, we conjecture that the equality is achieved. A similar situation for the single-shot regime is considered in the next example.

\subsubsection{The generalized erasure channel} \label{GEC}
 
We can define a generalization of the erasure channel by considering several output erasure symbols as follows. Let $\mathcal{P}=\{\mathcal{A}_{i}\}_{i=1}^{d}$ be a partition of $\mathcal{X}$. Define $P(Y|X)$ to be the channel such that $P(Y=x_i|X=x_i)=(1-\eta_i)$ and $P(Y=\alpha_{i}|X=x_i)=\eta_i$ for any $x_i \in \mathcal{A}_{i}$, with $0 \leq \eta_i \leq 1$, and $i=1,\cdots,d$. See Figure \ref{singleshot}. If $\min\{\eta_i^2\}_{i=1}^{d} \geq 1-\epsilon$, then the mapping $x_i \mapsto x_i^{*}$ is an $\epsilon$-reverse compression of the channel, with an arbitrarily fixed $x_i^{*}$ in $\mathcal{A}_i$. Its compressibility is equal to $\Gamma_{\epsilon}=(\mathcal{X}-d)/(\mathcal{X}-1)$.

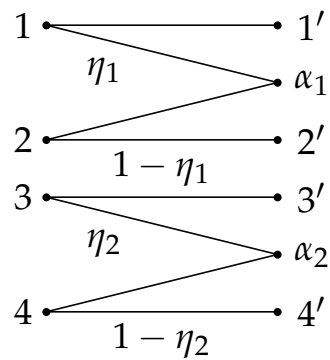
\begin{figure}   
\begin{adjustbox}{width=0.25\textwidth}
\begin{tikzpicture}
\draw (0,0)--(2,0); 
\draw (0,0)--(2,0.5);
\draw (0,1)--(2,1);
\draw (0,1)--(2,0.5); 
\node[scale=2] at (0,-0.02) {$\cdot$}; \node[scale=1] at (-0.2,0.0) {$4$};
\node[scale=2] at (0,0.98) {$\cdot$}; \node[scale=1] at (-0.2,1) {$3$};
\node[scale=2] at (2,0.98) {$\cdot$}; \node[scale=1] at (2.3,1) {$3'$};
\node[scale=2] at (2,0.48) {$\cdot$}; \node[scale=1] at (2.3,0.5) {$\alpha_2$};
\node[scale=2] at (2,-0.02) {$\cdot$}; \node[scale=1] at (2.3,0) {$4'$};

\draw (0,1.5)--(2,1.5); 
\draw (0,1.5)--(2,2);
\draw (0,2.5)--(2,2.5);
\draw (0,2.5)--(2,2); 
\node[scale=2] at (0,1.48) {$\cdot$}; \node[scale=1] at (-0.2,1.5) {$2$};
\node[scale=2] at (0,2.48) {$\cdot$}; \node[scale=1] at (-0.2,2.5) {$1$};
\node[scale=2] at (2,1.48) {$\cdot$}; \node[scale=1] at  (2.3,2) {$\alpha_1$};
\node[scale=2] at (2,2.48) {$\cdot$}; \node[scale=1] at  (2.3,2.5) {$1'$};
\node[scale=2] at (2,1.98) {$\cdot$}; \node[scale=1] at  (2.3,1.5) {$2'$}; 
\node[scale=1] at  (0.5,2.1) {$\eta_1$}; \node[scale=1] at  (0.5,0.6) {$\eta_2$};
\node[scale=1] at  (1,1.25) {$1-\eta_1$}; \node[scale=1] at  (1,-0.25) {$1-\eta_2$};

\end{tikzpicture}
\end{adjustbox}
    \caption{Generalized erasure channel. The picture depicts an example with $\mathcal{X}=\{1,2,3,4\}$, $\mathcal{Y}=\{1',2',3',4',\alpha_1,\alpha_2\}$. With probability $\eta_1$, $\{1,2\}$ are mapped to the erasure symbol $\alpha_1$, while they are unchanged with probability $1-\eta_1$. Similarly, $\eta_2$ is the probability for erasing $\{3,4\}$ to the distinct symbol  $\alpha_2 \neq \alpha_1$.}
      \label{singleshot}
\end{figure}

The next section considers the reverse compression method for quantum communication schemes.

\section{Quantum Reverse Compression} \label{quantum}

In order to extend the reverse compression method to quantum processes it is crucial to properly define the concept of partition for quantum systems, also often called coarse graining. The definition of quantum coarse graining to be considered here aims to deal with the problem of reverse compression, but it is possible to have alternative definitions depending on the problems addressed as it can be found for instance in \cite{busch1993concepts, brandes, walborn, petz2007quantum}.

We start by considering how the classical coarse graining associated with the reverse compression is represented as a CPTP map, i.e., a quantum channel. We then analyse the erasure channel, and confirm that the results remain valid using this framework. After this, we propose a natural generalization of quantum coarse graining, thus defining the reverse compression of a quantum channel.

A coarse graining should map the state space of a system into a state space of smaller dimension, and furthermore it must be regarded as a quantum channel $\Lambda^{X \rightarrow Z} \colon \mathcal{D}(\mathcal{H}_{X}) \rightarrow \mathcal{D}(\mathcal{H}_{Z})$ such that $\mathcal{H}_{Z}$ is a subspace with $\dim(\mathcal{H}_{Z}) \leq \dim(\mathcal{H}_{X})$. Denoting the correspondent classical partition by $\{ \mathcal{A}_z \}_{z\in \mathcal{A}_z}$, we define the projector on the subspace spanned by $\{\ket{x_{z}}\}_{x_{z}\in \mathcal{A}_z}$ as

\begin{equation}
\Pi_{z} \coloneqq \sum_{x_{z}\in \mathcal{A}_z} \left[ x_{z} \right],
\end{equation}
where $\left[ \psi \right]$ stands for the pure state $\ket{\psi} \bra{\psi}$.

The quantum operation corresponding to $\lambda:x_{z}\mapsto z$ is the measurement of a state $\rho$ in the subspace corresponding to $\mathcal{A}_z$ and the preparation of the pure state $\left[ z \right]$, represented by 

\begin{equation}
\Lambda_{z}^{X \rightarrow Z}(\rho) \coloneqq \left[ z \right] \Tr (\Pi_{z} \rho).
\end{equation}
It is a CP mapping $\mathcal{L}(\mathcal{H}_{X}) \rightarrow \mathcal{L}(\mathcal{H}_{Z})$ with decomposition $\Lambda_{z}^{X \rightarrow Z}(\rho)=\sum_{x \in \mathcal{A}_z} K_{x}^{(z)} \rho K_{x}^{(z)\dagger}$ given by Kraus operators defined as $K_{x}^{(z)} =\ket{z} \bra{x}$.

The quantum coarse graining corresponding to a classical partition is defined by summing over all outcomes

\begin{equation} \label{coarseop}
\Lambda^{X \rightarrow Z} \coloneqq \sum_{z \in \mathcal{Z}} \Lambda_{z}^{X \rightarrow Z}.
\end{equation}
It is clearly a CPTP map as its Kraus operators satisfy $\sum_{z \in \mathcal{Z}} \sum_{x \in \mathcal{A}_z} K_{x}^{(z)\dagger}K_{x}^{(z)}  = \mathbb{1}^X$, where $\mathbb{1}^X$ is the identity operator acting on $\mathcal{H}_{X}$.

\textit{Partial trace as a coarse graining.---}Consider the bipartite Hilbert space $\mathcal{H}_X=\mathcal{H}_Z \otimes \mathcal{H}_W$ and the coarse-graining 

\begin{equation}
\Lambda^{X \rightarrow Z}(\rho)=\Tr_{W}(\rho) = \sum_{w}\left( \mathbb{1}^Z \otimes \bra{w} \right) \rho \left( \mathbb{1}^Z \otimes \ket{w} \right), 
\end{equation}
where $\{\ket{w}\}_{w \in \mathcal{W}}$ is an orthonormal basis of $\mathcal{H}^W$. It corresponds to a classical partition of $\mathcal{X}=\mathcal{Z}\times\mathcal{W}$ by setting the elements $\mathcal{A}_z \coloneqq \{(z,w) \colon w \in \mathcal{W}\}$. That is, a uniform partition.
\newline

Similarly to its classical version, the quantum erasure channel outputs the input state $\rho$ with probability $1-\eta$ and outputs a pure {\it erasure state} $\left[ \alpha \right]$ with probability $\eta$. The erasure state is orthogonal to any input one, that is, $\Tr(\left[ \alpha \right]\rho)=0$. It means the input Hilbert space $\mathcal{H}_{X}$ must have smaller dimension than the output Hilbert space $\mathcal{H}_{Y}$, and more precisely $\mathcal{H}_{Y} \cong \mathbb{C}^{|\mathcal{X}|+1}$. This channel is denoted here as
\begin{equation}
\mathcal{M}_{\eta}^{X \rightarrow Y}(\rho)=(1-\eta)\rho + \eta \left[\alpha\right].
\end{equation}

Consider the case where the input states are taken to be pure states of the canonical basis. The reverse compression here consists in defining the optimal partition such that each pair of pure states $\left[x_z\right]$ and $\left[\widehat{x}_z\right]$ satisfies the condition
\begin{eqnarray}
&\widetilde{F}&\left(\left[x_z\right],\left[\widehat{x}_z\right]\right) \coloneqq \nonumber \\
&F&\left(\mathcal{M}_{\eta}^{X \rightarrow Y} \left(\left[x_z\right]\right),\mathcal{M}_{\eta}^{X \rightarrow Y}\left( \left[\widehat{x}_z\right]\right)  \right) \geq 1-\epsilon.
\end{eqnarray}
As the reverse fidelity $\widetilde{F}$ of distinct pure states $\left[x\right]$ and $\left[\widehat{x}\right]$ with respect to the erasure channel is equal to $\eta^{2}$, then the quantum reverse compression goes in the same lines as its classical counterpart: it is either possible to perfectly (reversely) compress every pure state in the canonical basis to the same fixed pure state if $\eta^2 \geq 1-\epsilon $, or it is not possible to (reversely) compress at all if $\eta^2 < 1-\epsilon $.

The coarse graining operation defined in Equation (\ref{coarseop}) maps distinct collections of pure orthogonal states into distinct fixed orthogonal pure states. We can consider more general coarse grainings, without a direct classical analogue, that do not rely on a particular basis choice. In order to define reverse compression for such cases, we introduce the vectorial kernel of a linear map $\Lambda^{X \rightarrow X'}:  \mathcal{L}(\mathcal{H}_{X}) \rightarrow \mathcal{L}(\mathcal{H}_{X'})$, with ${H}_{X} \cong {H}_{X'}$ as the set

\begin{align}
\mathcal{K}\left(\Lambda^{X \rightarrow X'}\right)& \coloneqq \{ \ket{\psi} \in \mathcal{H}_{X'} : \nonumber \\ &\Lambda^{X \rightarrow X'}(\rho)\ket{\psi}=0, \, \forall  \rho \in \mathcal{D}(\mathcal{H}_{X}) \}.
\end{align}
This is clearly a subspace of the associated channel's input vector space. More precisely it is a subspace of the output space of the coarse graining operation, which is defined to be isomorphic to the input one. The range of the map $\Lambda^{X \rightarrow X'}$,

\begin{align}
\mathcal{R}(\Lambda^{X \rightarrow X'}) &\coloneqq  \nonumber \\ &\{ \Lambda^{X \rightarrow X'}(\rho) \in \mathcal{L}(\mathcal{H}_{X'}) \, : \, \rho \in \mathcal{L}(\mathcal{H}_{X}) \},
\end{align}
is equal to the set of linear operators acting on the orthogonal complement of the vector kernel of the map, $\mathcal{R}(\Lambda^{X \rightarrow X'})=\mathcal{L}(\mathcal{K}(\Lambda^{X \rightarrow X'})^{\bot})$, with $\mathcal{K}(\Lambda^{X \rightarrow X'})^{\bot}\coloneqq \{ \ket{\phi} \in \mathcal{H}^{X'}:\bra{\phi}\ket{\psi}=0, \forall \ket{\psi} \in \mathcal{K}(\Lambda^{X \rightarrow X'}) \}$. The range of a quantum coarse graining plays the role of the classical compressed input space $\mathcal{Z}$, in the sense that, after coarse graining, only states from within the range are sent through the channel.

We say channels $\mathcal{R}^{X \rightarrow Y}$ and $\mathcal{T}^{X \rightarrow Y}$ are $\epsilon$-indistinguishable, and denote it as $\mathcal{R}^{X \rightarrow Y} \myeq \mathcal{T}^{X \rightarrow Y}$, if and only if they satisfy $F[\mathcal{R}^{X \rightarrow Y}(\rho),\mathcal{T}^{X \rightarrow Y}(\rho)] \geq 1-\epsilon $ for any state $\rho$ in $\mathcal{H}_{X}$. The reverse compression method for an arbitrary channel in the most general setting is defined as follows.

\begin{definition}[Single-Shot Quantum Reverse Compression] \label{quantumcompression}
An $\epsilon$-reverse compression of the quantum channel $\mathcal{N}^{X\rightarrow Y}$ is a quantum channel $\Lambda^{X \rightarrow X'}$ such that 

\begin{equation}
\mathcal{N}^{X' \rightarrow Y} \circ \Lambda^{X \rightarrow X'}  \myeq \mathcal{N}^{X \rightarrow Y},
\end{equation}
with maximum quantum compressibility

\begin{equation}
\Gamma_{\epsilon}(\mathcal{M}^{X\rightarrow Y})\coloneqq \frac{\dim \mathcal{K}(\Lambda^{X \rightarrow X'}) }{\dim \mathcal{H}^{X}-1}.
\end{equation}

\end{definition}

A non-trivial compression is identified as a non-full range CPTP map $\Lambda^{X \rightarrow X'}:\mathcal{L}(\mathcal{H}_{X}) \rightarrow \mathcal{L}(\mathcal{H}_{X'})$. Considering this most general definition, the quantum reverse compression is clearly a superchannel acting on the channel of interest. A superchannel is a linear transformation of quantum channels into quantum channels, with possibly different dimensions of input and output spaces. A transformation of this type (for the case of single-party) occur in general as an isometric preprocessing and postprocessing with a side identity channel \cite{chiribella2008transforming}. In what follows, we show that, even allowing for arbitrary quantum coarse grainings, the reverse compression holds the same result for the erasure channel.

\begin{theorem}
The quantum erasure channel can be reversely compressed if and only if $\eta^2 \geq 1-\epsilon $.
\end{theorem}

\begin{proof}

\textit{Direct part.---}If $\eta^2 \geq 1-\epsilon $, then the input space can be perfectly compressed to a fixed state. It can be easily checked from the reverse fidelity for arbitrary input states as

\begin{eqnarray}
&\widetilde{F}&\left[\rho,\Lambda^{X \rightarrow X'}(\rho) \right] \coloneqq \nonumber \\
&F&\left[\mathcal{M}_{\eta}^{X \rightarrow Y}(\rho),\mathcal{M}_{\eta}^{X' \rightarrow Y}(\Lambda^{X \rightarrow X'}(\rho)) \right] = \nonumber \\
&& \left[ (1-\eta)\sqrt{F\left[\rho,\Lambda^{X \rightarrow X'}(\rho)\right]} + \eta \right]^2 \geq \eta^2.
\label{erasurefidelity}
\end{eqnarray}
As a consequence, we can reversely compress all the input space to an arbitrary pure input state according to $\Lambda^{X \rightarrow X'}(\rho)=[1]$. This clearly satisfies the threshold $\widetilde{F}\left[\rho,\Lambda^{X \rightarrow X'}(\rho) \right]\geq 1-\epsilon$, and achieves the highest possible compressibility $\Gamma_{\epsilon}(\mathcal{M}_{\eta}^{X \rightarrow Y})=1$.
\newline 

\textit{Converse part.---}If $\eta^2 < 1-\epsilon $, then no reverse compression is possible -- meaning that any reverse compression $\Lambda^{X \rightarrow X'}$ would be a full range map. In order to check this, suppose there is a non-trivial vector in the vector kernel of the map, $\ket{\psi} \in \mathcal{K}(\Lambda^{X \rightarrow X'})$. In such a case it holds that $F\left([\psi],\Lambda^{X \rightarrow X'}([\psi]) \right)=0$, and therefore, using the expression \eqref{erasurefidelity} for the fidelity of the erasure channel, we find $F\left[\mathcal{M}_{\eta}^{X \rightarrow Y}([\psi]),\mathcal{M}_{\eta}^{X' \rightarrow Y}(\Lambda^{X \rightarrow X'}([\psi])) \right] = \eta^2 < 1-\epsilon $. We conclude that, for $\eta^2 < 1-\epsilon$, the reverse compression map must have a trivial vector kernel, achieving compressibility $\Gamma_{\epsilon}(\mathcal{M}_{\eta}^{X \rightarrow Y})=0$.
\end{proof}

As it was the case for the classical reverse compression, we can easily extend the quantum reverse compression to the asymptotic regime by considering its action on several copies of the input space as follows.

\begin{definition}[Asymptotic Quantum Reverse Compression]
An asymptotic quantum $\epsilon$-reverse compression of a channel $\mathcal{N}^{X\rightarrow Y}$ is a sequence of quantum channels $\left( \Lambda_{k}^{X^k\rightarrow X'^k}: \mathcal{L}\left(\mathcal{H}_{X}^{\otimes k}\right) \rightarrow \mathcal{L}\left(\mathcal{H}_{X'}^{\otimes k}\right) \right)_{k \in \mathbb{N}}$ such that each channel $\Lambda_{k}$ is an $\epsilon$-reverse compression of $k$ independent copies of the channel $\mathcal{N}$, with compressibility $\Gamma_{\epsilon}^{(k)}(\mathcal{N})\coloneqq \Gamma_{\epsilon}\left(\mathcal{N}^{\otimes k}\right)$. The asymptotic compressibility is defined as the limit $\Delta_{\epsilon}(\mathcal{N}^{X\rightarrow Y})\coloneqq \lim_{k\rightarrow \infty} \Gamma_{\epsilon}^{(k)}\left(\mathcal{N}\right)$.

\end{definition}

For the case of separable input states the asymptotic quantum reverse compression works similarly to its classical counterpart. The more general case of entangled input states is left for future studies.

\section{Conclusion} \label{conclusion}

We have considered here the task of preprocessing classical and quantum systems according to a method we called reverse compression. It can be understood as complementary to Shannon compression. While the latter is irrelevant for sources with uniform probability distributions, the former is irrelevant for identity channels, the case where distinct data is perfectly distinguishable for the receiver.  Additionally, we have considered the erasure channel as a study case. This presents an all vs nothing situation where, depending on the amount of noise tolerated, it is either possible to perfectly compress the input of the channel or it is not possible at all. This happens because the noise in the erasure channel does not depend upon the input state. A more detailed analysis of asymptotic settings is left for future studies as well. Moreover, the study of different preprocessing operations, and general channel simulation tasks is an open area.


\begin{acknowledgments}

F.C.\ acknowledges support through an Australian Research Council Discovery Early Career Researcher Award (DE170100712). M.C. acknowledges support from the the funding agency CNPq.  M.C. acknowledges the warmly hospitality of the School of Mathematics and Physics at the University of Queensland, and also thanks to Lucas Chibebe C\'eleri for useful discussions. This study was financed in part by the Coordena\c c\~ ao de Aperfei\c coamento de Pessoal de N\'ivel Superior - Brasil (CAPES) - Finance Code 001. We acknowledge the traditional owners of the land on which the University of Queensland is situated, the Turrbal and Jagera people.
\end{acknowledgments}


\newpage

\begin{widetext}
\appendix

\subsection{Reverse fidelity of memoryless channels used without feedback} \label{fidelity}

The fidelity of conditional probabilities distributions $P(Y^k|x^k)=\prod_{i=1}^{k} P(Y_i|x_i)$ and $P(Y^k|\widehat{x}^k)=\prod_{i=1}^{k} P(Y_i|\widehat{x}_i)$ can be written as $F[P(Y^k|x^k),P(Y^k|\widehat{x}^k)] = \prod_{i=1}^{k} F[P(Y_i|x_i), P(Y_i|\widehat{x}_i)]$. The proof follows as

\begin{eqnarray} \label{fidasymp}
F[P(Y^k|x^k),P(Y^k|\widehat{x}^k)] &=& \left[ \sum_{y^k} \sqrt{P(y^k|x^k) P(y^k|\widehat{x}^k)} \right]^2 \nonumber \\
                         &=& \left[ \sum_{y^k} \left[\prod_{i=1}^{k} P(y_i|x_i)  P(y_i|\widehat{x}_i)\right]^{\frac{1}{2}} \right]^2 \nonumber \\
                         &=& \left[ \sum_{y^k} \prod_{i=1}^{k}  \sqrt{P(y_i|x_i)  P(y_i|\widehat{x}_i)} \right]^2 \nonumber \\
                         &=&  \prod_{i=1}^{k} \left[ \sum_{y_i}  \sqrt{P(y_i|x_i)  P(y_i|\widehat{x}_i)} \right]^2 \nonumber \\
                         &=&  \prod_{i=1}^{k} F\left[P(y_i|x_i),P(y_i|\widehat{x}_i) \right]. 
\end{eqnarray}

A quantum version of Equation (\ref{fidasymp}) can be achieved by considering the action of $k$ independent copies of a channel $\mathcal{N}^{X^{k} \rightarrow Y^{k}}=\mathcal{N}^{X_{1} \rightarrow Y_{1}}\otimes \cdots \otimes \mathcal{N}^{X_{k} \rightarrow Y_{k}}$ on $k$ i.i.d copies of states $\rho$ and $\sigma$, that is, $\rho^{X^k}=\rho^{X_1}\otimes \cdots \otimes \rho^{X_k}$ and $\sigma^{X^k}=\sigma^{X_1}\otimes \cdots \otimes \sigma^{X_k}$. It follows that

\begin{eqnarray}
F\left[\mathcal{N}^{X^{k} \rightarrow Y^{k}}(\rho^{X^k}),\mathcal{N}^{X^{k} \rightarrow Y^{k}}(\sigma^{X^k})\right] &=& \left[ \Tr_{Y^k} \left[ \sqrt{\sqrt{\mathcal{N}^{X^{k} \rightarrow Y^{k}}(\rho^{X^k})}\mathcal{N}^{X^{k} \rightarrow Y^{k}}(\sigma^{X^k})\sqrt{\mathcal{N}^{X^{k} \rightarrow Y^{k}}(\rho^{X^k})}} \right] \right]^{2} \nonumber \\
					  &=& \prod_{i=1}^k \left[\Tr_{Y_i}\left[ \sqrt{\sqrt{\mathcal{N}^{X_{i} \rightarrow Y_{i}}(\rho^{X_i})}\mathcal{N}^{X_{i} \rightarrow Y_{i}}(\sigma^{X_i})\sqrt{\mathcal{N}^{X_{i} \rightarrow Y_{i}}(\rho^{X_i})}} \right] \right]^2	\nonumber \\
					  &=& \prod_{i=1}^k F\left[\mathcal{N}^{X_{i} \rightarrow Y_{i}}(\rho^{X_i}),\mathcal{N}^{X_{i} \rightarrow Y_{i}}(\sigma^{X_i})\right] .
\end{eqnarray}

\subsection{Reverse fidelity of the erasure channel} \label{reversefidEC}

We present here a detailed derivation of the reverse fidelity of an erasure channel. Recall that the erasure channel is described by the conditional probability distribution $P(y|x)=\left(1-\eta \right)\delta_{x,y}+\eta \delta_{\alpha,y}$, with $x \in \mathcal{X}\coloneqq\{1,\cdots,r\}$, $y \in \mathcal{Y}\coloneqq \{1,\cdots,r,\alpha\}$, and $0 \leq \eta \leq 1$. It follows then

\begin{eqnarray} \label{revfid}
\widetilde{F}(x,\widehat{x}) &\triangleq & F\left[ P(Y|x), P(Y|\widehat{x}) \right] \nonumber \\
							 &=& \left[  \sum_{y\in \mathcal{Y}} \sqrt{\left[\left(1-\eta \right)\delta_{x,y}+\eta \delta_{\alpha,y}\right]\left[\left(1-\eta \right)\delta_{\widehat{x},y}+\eta \delta_{\alpha,y}\right]} \right]^{2}  \nonumber \\
							 &=& \left[ \sum_{y\in \mathcal{Y}} \sqrt{\left(1-\eta \right)^{2} \delta_{x,y} \delta_{\widehat{x},y} + \left(1-\eta \right)\eta \delta_{x,y} \delta_{\alpha,y} + \eta \left(1-\eta \right)\delta_{\alpha,y} \delta_{\widehat{x},y} + (\eta \delta_{\alpha,y})^{2}}  \nonumber  \right]^{2} \\
							 &=& \begin{cases}
                       				 1, \qquad  \,\,\, \text{$x=\widehat{x}$} \\
                        			 \eta^{2}, \qquad \text{$x\neq\widehat{x}$}
                    			 \end{cases}
\end{eqnarray}

Considering Equation (\ref{revfid}) we get the reverse fidelity of a memoryless erasure channel used without feedback

\begin{eqnarray}
\widetilde{F}(x^k,\widehat{x}^k) &\triangleq & F\left[ P(Y^k|x^k), P(Y^k|\widehat{x}^k) \right] \nonumber \\
								 &=& \prod_{i=1}^{k} F\left[ P(Y_i|x_i), P(Y_i|\widehat{x}_i) \right] \nonumber \\	
								 &=& \left( \eta^{2} \right)^{S(x^k,\widehat{x}^k)}
\end{eqnarray}
where $S(x^k,\widehat{x}^k)$ is the number of different components of $x$ and $\widehat{x}^k$.

\subsection{On the conjecture considered in section \ref{examples}} \label{Conjecture}

We consider here the asymptotic reverse compression of channels, allowing for some error $\epsilon$ in fidelity, such that the elements of the associated partition $\mathcal{P}_{\epsilon}^{(k)}$ cannot have more than $s$ different components. First we show the validity of the upper bound 
\begin{equation}
|\mathcal{P}_{\epsilon}^{(k)} | \leq |\mathcal{X}|^{k-s}.
\label{upperbound}
\end{equation}

It can be proved by building a partition with such properties and with cardinality equal to $|\mathcal{X}|^{k-s}$. As an example, we consider binary sequences ($\mathcal{X}=\{0,1\}$) of length $k=3$. All possible sequences can be written in lexicographical order as follows: $000$, $001$, $010$, $011$, $100$, $101$, $110$ and $111$. If $s=0$ there is nothing to be done as the only possible partition is the trivial $\{\{x^3\}\}_{x^3}$. Therefore, in this case we can only have $|\mathcal{P}_{\epsilon}^{(3)} | = 8 = |\mathcal{X}|^{3}$. For $s=1$, we can merge two sequences together as in the partition $\{\mathcal{A}_1, \mathcal{A}_2, \mathcal{A}_3, \mathcal{A}_4 \}$, where $\mathcal{A}_1=\{000, 001\}$, $\mathcal{A}_2=\{010, 011\}$, $\mathcal{A}_3=\{100, 101\}$ and $\mathcal{A}_4=\{110 ,111\}$. This gives us $|\mathcal{P}_{\epsilon}^{(3)} | = 4 = |\mathcal{X}|^{3-1}$. For $s=2$, it is possible to define a partition following the same reasoning, $\{\mathcal{A}_1 \cup \mathcal{A}_2 = \{000, 001, 010, 011\}, \mathcal{A}_3 \cup \mathcal{A}_4 = \{100, 101, 110 ,111\}\}$, and for $s=3$ all the sequences can be merged together generating the partition $\{\mathcal{A}_1 \cup \mathcal{A}_2 \cup \mathcal{A}_3 \cup \mathcal{A}_4\}$, in both cases saturating the bound \eqref{upperbound}. This procedure can clearly be extended to arbitrary $k$-ary input spaces, and there is a partition $\mathcal{Q}_{s}^{(k)}$ with the desired properties, and such that $|\mathcal{Q}_{\epsilon}^{(k)} | = |\mathcal{X}|^{k-s}$.

For the preceding example it can be shown there is no partition with the defining properties such that it has fewer than $|\mathcal{X}|^{k-s}$ elements. For $s=0$ and $s=3$, the proof is trivial. In order to prove that for $s=2$, define the partition of $\{0,1\}^3$ given by sets $\mathcal{B}_1=\{000,111\}$, $\mathcal{B}_2=\{001,110\}$, $\mathcal{B}_3=\{010,101\}$ and $\mathcal{B}_4=\{011,100\}$. Such partition is defined in a way that each element $\mathcal{B}_j$ consists of sequences with $s = 3$ different components. Therefore, the desired assertion follows by the pigeonhole principle for $s=2$ as each $\mathcal{B}_j$ contains two strings. That is, we can at most take a partition with two elements $\mathcal{A}_1$ and $\mathcal{A}_2$, such that each $\mathcal{A}_i$ ($i=1,2$) contains at most one string in $\mathcal{B}_j$ ($j=1,\cdots,4$). For instance, one such partition could be with elements $\mathcal{A}_1=\{000,001,010,100\}$, and $\mathcal{A}_1=\{111,110,101,011\}$. For $s=1$, we can prove it by considering the partition with sets $\mathcal{B}_1=\{000,011,110,101\}$ and $\mathcal{B}_1=\{111,100,001,010\}$. It is defined in such a way that $x,y \in \mathcal{B}_j$ ($j=1,2$), inplies $S(x,y)\geq 2$. Therefore, the proof follows similarly by the pigeonhole principle as each $\mathcal{B}_j$ contains 4 strings. We refer the reader to \cite{merris} for a clear presentation on the application of the pigeonhole principle to problems in combinatorics. Note that for $k=4$ it is already not so simple to show that, as the same line of reasoning does not hold. By this reason, we conjecture that the equality must hold for any value of $k$.

\subsection{Asymptotic compressibility of the generalized eresure channel} \label{ACGEC}

We show here that the generalized erasure channel defined in \ref{GEC} has zero asymptotic compressibility. It can be seen by exploiting its definition that $\widetilde{F}(x_i,x_j)=0$, for $x_i \in \mathcal{A}_i$ and $x_j \in \mathcal{A}_j$, with $i \neq j$. Therefore, if $\epsilon < 1$ and one wants to reversely compress the input sequences of $\mathcal{X}^k$, one cannot map sequences with components belonging to different sets $\mathcal{A}_i$ to the same compressed one. Otherwise, it would lead to zero reverse fidelity. Then, the best we can do is to compress the sequences with all the components in the same set $\mathcal{A}_i$ to the same sequence. That means we are exploiting the partition 

\begin{equation}
\mathcal{X}^k = \left( \bigcup_{i=1}^{d} \mathcal{A}_i^{k} \right) \cup \mathcal{B},
\end{equation}
where $\mathcal{B}$ contains all the sequences which is not in any of the sets $\mathcal{A}_i^{k}$. Precisely, the partition $\mathcal{P}_{\epsilon}^{(k)}=\{\mathcal{A}_1^{k},\cdots,\mathcal{A}_d^{k},\mathcal{B}\}$ has cardinality $|\mathcal{P}_{\epsilon}^{(k)}|=d+|\mathcal{B}|$. 

Note that $\mathcal{X}= \cup_{i=1}^{d} \mathcal{A}_i$ and $|\mathcal{X}|=\sum_{i=1}^{d} |\mathcal{A}_i|$. Therefore, it follows that

\begin{equation}
|\mathcal{X}^k| = \left( \sum_{i=1}^{d} |\mathcal{A}_i| \right)^k = \sum_{i=1}^{d} |\mathcal{A}_i|^k + |\mathcal{B}|,
\end{equation}
and we must have 
\begin{equation}
|\mathcal{B}|= \left( \sum_{i=1}^{d} |\mathcal{A}_i| \right)^k - \sum_{i=1}^{d} |\mathcal{A}_i|^k.
\end{equation}

By definition, the compressibility of $k$ independent uses of the generalized erasure channel is upper bounded by

\begin{equation}
	\Gamma_{\epsilon}^{(k)} \leq \frac{\sum_{i=1}^{d} |\mathcal{A}_i|^k-d}{\left( \sum_{i=1}^{d} |\mathcal{A}_i| \right)^k-1}.
\end{equation}
Moreover, if $\epsilon < 1$, then the asymptotic compressibility of the generalized erasure channel is $\Delta_{\epsilon}=0$.
\end{widetext}
\end{document}